\documentclass{aims}
\usepackage{amsmath}
  \usepackage{paralist}
  \usepackage{graphics} 
  \usepackage{epsfig} 
\usepackage{graphicx}  \usepackage{epstopdf}
 \usepackage[colorlinks=true]{hyperref}
\hypersetup{urlcolor=blue, citecolor=red}

\def\currentvolume{X}
 \def\currentissue{X}
  \def\currentyear{200X}
   \def\currentmonth{XX}
    \def\ppages{X--XX}

\newtheorem{theorem}{Theorem}[section]
\newtheorem{corollary}{Corollary}

\newtheorem{lemma}[theorem]{Lemma}
\newtheorem{proposition}{Proposition}

\theoremstyle{definition}
\newtheorem{definition}[theorem]{Definition}

\title[Bias of Nonnegative Mechanisms] 
      {Observations on the Bias of Nonnegative Mechanisms for Differential Privacy}

\author[A. McGlinchey and O. Mason]{}


 \email{oliver.mason@mu.ie}
 \email{aisling.mcglinchey@mu.ie}

\thanks{This work is supported by SFI grant 13/RC/2094}

\thanks{$^*$ Corresponding author: O. Mason}

\begin{document}
\maketitle

\centerline{\scshape Aisling McGlinchey}
\medskip
{\footnotesize
 \centerline{Dept. of Mathematics and Statistics,}
   \centerline{ Maynooth University, Co. Kildare, Ireland}
   \centerline{\& Lero, the Science Foundation Ireland Research Centre for Software}
} 

\medskip

\centerline{\scshape Oliver Mason$^*$}
\medskip
{\footnotesize
 \centerline{ Dept. of Mathematics and Statistics,}
   \centerline{ Maynooth University, Co. Kildare, Ireland}
   \centerline{\& Lero, the Science Foundation Ireland Research Centre for Software}
}

\bigskip

 \centerline{(Communicated by the associate editor name)}

\begin{abstract}
We study two methods for differentially private analysis of bounded data and extend these to nonnegative queries.  We first recall that for the Laplace mechanism, boundary inflated truncation (BIT) applied to nonnegative queries and truncation both lead to strictly positive bias.  We then consider a generalization of BIT using translated ramp functions.  We explicitly characterise the optimal function in this class for worst case bias.  We show that applying any square-integrable post-processing function to a Laplace mechanism leads to a strictly positive maximal absolute bias.  A corresponding result is also shown for a generalisation of truncation, which we refer to as restriction.  We also briefly consider an alternative approach based on multiplicative mechanisms for positive data and show that, without additional restrictions, these mechanisms can lead to infinite bias.    
\end{abstract}

\section{Introduction}
\label{sec:Intro}

There has been substantial interest in the past two decades in the development and analysis of formal privacy frameworks for the release and analysis of personal data \cite{DworkRoth14, Torra17}.  Differential privacy (DP) \cite{Dwork06A} provides very strong, provable privacy guarantees and is one of the most important privacy frameworks from both a practical and theoretical standpoint.  Since its introduction, many results have been derived on the fundamental theory of DP and its relation to other privacy concepts (see for instance \cite{DworkRoth14, HalRinWas12, SorDom13, DomSor15, SorDomSanMeg17, KalSanSar18}).  Some of the most important differentially private mechanisms  proposed in the literature include the Laplace,  Gaussian, and Exponential mechanisms \cite{SheTal07}.  Numerous other mechanisms have also been proposed, including, for example, the discrete Laplace \cite{SorDomSanMeg17}, Staircase, and Generalized Gaussian mechanisms \cite{Liu19}.

Since its introduction, differential privacy has been used in a diverse range of applications.  For example, \cite{LeNyPappas13} describes applications to control theory while \cite{HolLeiMas17} considers differential privacy for randomized response surveys.  A high-profile and important application of differential privacy that is relevant to the problems considered here is its use for census data.  Recently, the US Census Bureau announced a plan to use differential privacy for the 2020 census results \cite{Abowd}.  Several researchers have noted specific challenges that arise in the use of differential privacy in such contexts.  In particular, the nonnegative, discrete, nature of the count data in a census makes standard mechanisms unsuitable and requires either new mechanisms or appropriate ways of adapting existing ones \cite{Cal20, FioHen20A}.  It is fundamentally important that mechanisms for constrained data (such as integer-valued, nonnegative or hierarchical data) respect these constraints; otherwise the outputs from the mechanisms may be unrealistic.  The recent paper \cite{Cal20} presents a novel mechanism for count queries on census-type data.  The framework developed is for approximate or relaxed $(\varepsilon, \delta)$-differential privacy and can be used to design integer-valued mechanisms with specified range and error probability.  A related line of work in \cite{FioHen20A} describes an optimization-based approach using the geometric distribution that respects integral as well as hierarchical constraints between groups in the dataset.  Both of these approaches are suitable for count queries arising in connection with census data.

In this paper, our interest is in adapting existing mechanisms so that they respect nonnegativity constraints.  We focus on two methods, post-processing and restriction (with the latter being equivalent to rejection sampling), for the, admittedly limited, but practically important case of the Laplace mechanism.  This is arguably the most widely used and studied mechanism for $\varepsilon$-differential privacy and real-valued queries; according to the authors of \cite{HolAntBraAon18}, it is the ``workhorse of differential privacy''.  The range of Laplace random variables is $\mathbb{R}$, making it unsuitable for queries and data subject to constraints.  For bounded queries, taking values in some interval $[l, u] \subset \mathbb{R}$ for example, it is necessary to adapt the mechanism to ensure realistic, meaningful outputs.  Recent work describing ways of addressing this issue for bounded queries can be found in \cite{HolAntBraAon18, Liu16}.  

Two approaches to construct bounded $\varepsilon$-differentially private mechanisms from the Laplace mechanism are described in \cite{Liu16}: namely, \emph{truncation} and \emph{boundary inflated truncation}.  Boundary inflated truncation works by setting values outside the query range to the nearest boundary value; this is an application of the general principle that differential privacy is not affected by \emph{post-processing} with deterministic functions \cite{DworkRoth14}.  Truncation resamples from the Laplace distribution until a value in the specified range is obtained.  As noted in \cite{Liu16}, this is equivalent to a form of \emph{rejection sampling}.  These methods were also studied for the Generalized Gaussian mechanism in \cite{Liu19}. In \cite{HolAntBraAon18}, the authors studied truncated Laplace mechanisms (they refer to them as bounded Laplace mechanisms) and derived conditions for these mechanisms to satisfy relaxed $(\varepsilon, \delta)$ differential privacy.  The bias formulae derived for truncation and boundary inflated truncation in \cite{Liu16} motivate our work here.

Our motivation stems from control theory \cite{LeNyPappas13}, in particular so-called positive systems \cite{ValRan2018} which arise in applications such as population dynamics and transportation.  For this reason, we focus on constructions for nonnegative-valued queries.  Moreover, with the application to control systems in mind, we do not assume a priori that the data or query outputs are bounded.  This contrasts with the recent, interesting, results in \cite{FioHen20B} motivated by the hierarchical constraints arising in applications to census data.  This latter paper studies the effect of post-processing for queries constrained to bounded feasible sets described by linear equations.  It is shown that when a particular class of post-processing functions, known as projections, is applied to the Laplace mechanism, bias is specifically the result of the addition of a nonnegativity constraint.  While outside the scope of this paper, it is important to note that there are other possible approaches to the problem of private data analysis on nonnegative or otherwise constrained data.  In particular, as suggested by one of the anonymous reviewers of this paper, a Bayesian framework \cite{Bayes1} could be used to construct posterior distributions for the query output based on the noisy outputs and the analyst's knowledge of the mechanism used and data constraints.      

\emph{Summary of Contributions}

The contributions of the paper are qualitative in nature and concern the bias of nonnegative post-processed Laplace mechanisms and general restricted mechanisms.  Our two main results, Proposition \ref{prop:infbiaspos} and Proposition \ref{prop:biasinevrest} establish a fundamental limitation for post-processing and restriction by showing that bias is unavoidable for both approaches.  

In Section \ref{sec:bias1}, we present simple adaptations of the formulae for bias and MSE from \cite{Liu16} for queries taking values in $[0, \infty)$ and briefly consider a type of multiplicative mechanism introduced in \cite{LeNyPappas13}, for strictly positive queries.  We show that without additional restrictions, these mechanisms can have infinite worst case bias. 

The simplest approach to post-processing for nonnegative queries is to post-process with the standard ramp function.  In Section \ref{sec:BiasOpt}, we demonstrate that nonnegative mechanisms with improved worst-case bias can be obtained using alternative post-processing functions; we show this by considering the simple case of translated ramp functions and determine the optimal such function for worst case bias.  We prove that positive bias is inevitable for  any post-processed Laplace mechanism in Proposition \ref{prop:biasinevpp} and prove the stronger result that the worst case bias \emph{is bounded away from zero} for any such mechanism in Proposition \ref{prop:infbiaspos}.  These initial results suggest several directions for future research.  In particular, the problems of determining the optimal post-processing function and, more generally, the optimal mechanism for nonnegative queries, for the worst case bias are interesting theoretically and practically.    

We prove a similar fundamental limitation for restricted mechanisms in Proposition \ref{prop:biasinevrest}.  Here we show that for \emph{any initial mechanism} (not necessarily the Laplace mechanism) the maximal absolute bias of the restricted mechanism is again positive.  Taken together, these results establish important, fundamental properties and limitations of the generalizations of truncated and boundary inflated truncated mechanisms studied here.  From a practical point of view, they show that if an unbiased nonnegative mechanism is required, then alternative approaches need to be considered.  They also suggest several interesting directions for future research, some of which we describe in our concluding remarks.

\section{Background and Notation}
\label{sec:bkgd}
We briefly recall some key definitions and results on differential privacy; broadly, we follow the formalism in the paper \cite{HolLeiMas14}.  $(\Omega, \mathcal{F}, \mathbb{P})$ denotes a probability space where $\mathcal{F}$ is a $\sigma$-algebra of subsets of $\Omega$ and $\mathbb{P}$ is a probability measure on $\Omega$.  For a real-valued random variable $X:\Omega \to \mathbb{R}$, $\mathbb{E}[X]$ denotes its expectation.  

$D$ is a set representing the possible databases of interest and we assume it is equipped with a symmetric, reflexive (but not transitive) adjacency relation $\sim$; this defines the notion of similarity for databases.  A (real-valued) \textbf{query} is a mapping $Q:D \to \mathbb{R}$.  Given a query $Q$, a \emph{mechanism} is a collection of random variables $\{X_{Q, d}: d \in D\}$ where $X_{Q, d}:\Omega \to \mathbb{R}$ is a real-valued random variable for each $d \in D$.  In a slight abuse of notation, we shall often refer to the mechanism $X_{Q, d}$.  

We only consider output perturbation mechanisms here: if $Q$ has range $Q(D)$, such a mechanism is defined by a family $\{Y_q : q \in Q(D)\}$ of random variables $Y_q:\Omega \to \mathbb{R}$.  For $d \in D$, we simply set $X_{Q, d} = Y_{Q(d)}$.  

Given $\varepsilon > 0$, the mechanism $X_{Q,d}$ is $\varepsilon$-differentially private if 
\begin{equation}
\label{eq:DP}
\mathbb{P}(X_{Q,d} \in A) \leq e^{\varepsilon}\mathbb{P}(X_{Q,d'} \in A)
\end{equation} 
for all $d \sim d'$ in $D$ and all Borel sets $A \subseteq \mathbb{R}$.

It is well known \cite{DworkRoth14} that for any measurable, deterministic, function $\phi:\mathbb{R} \to \mathbb{R}$, if $X_{Q, d}$ is $\varepsilon$-differentially private, then so is the post-processed mechanism $\hat{X}_{Q, d} = \phi(X_{Q, d})$. 

The \textbf{sensitivity} of the query $Q:D \to \mathbb{R}$ is given by:
\[\Delta(Q) := \sup\{|Q(d) - Q(d')| : d, d' \in D, \, d \sim d'\}.\]

Recall, that a Laplace random variable with mean $q \in \mathbb{R}$ and scale parameter $b > 0$ is defined by the probability density function (pdf): 
\begin{equation}\label{eq:Lap}
f_q(x) = \frac{1}{2b} e^{-\frac{|x-q|}{b}}, \quad x \in \mathbb{R}.
\end{equation}
We shall use $L_b$ to denote a Laplace random variable with mean $0$ and scale parameter $b$.  The following result concerning the Laplace mechanism and differential privacy is well known; see \cite{Dwork06}, \cite{DworkRoth14}.

\begin{proposition}\label{prop:LapBasic}
Let $Q:D \to \mathbb{R}$ have sensitivity $\Delta$ and $\varepsilon > 0$ be given.  The \emph{Laplace mechanism} defined by $X_{Q, d} = Q(d) + L_b$ where $b \geq \frac{\Delta}{\varepsilon}$ is $\varepsilon$-differentially private.  
\end{proposition}
The Laplace mechanism is determined by the family of random variables $\{Y_q = q + L_b: q \in Q(D)\}$.  This notation (for the case where $Q(D) = [0, \infty)$) will be used frequently in the paper.  Note that $q$ denotes the 'true' query response and $\mathbb{E}[Y_q] = q$ so $Y_q$ is \emph{unbiased}.

Throughout the paper, we are concerned with the following general setup.  We are given a nonnegative valued query $Q:D \to [0, \infty)$ with sensitivity $\Delta$, and a mechanism $X_{Q, d} = Y_{Q(d)}$ where each random variable $Y_q$ has range $\mathbb{R}$ and $\mathbb{E}[Y_q] = q$.  We assume $X_{Q, d}$ is $\varepsilon$-differentially private.  For the majority of the paper, $Y_q$ will be a Laplace random variable.  We study the bias properties of two generalisations of boundary inflated truncation and truncation for constructing nonnegative mechanisms $\hat{X}_{Q, d} = \hat{Y}_{Q(d)}$.  Essentially, we construct a nonnegative (derived) family $\hat{Y}_q$ from the given family $Y_q$ such that the associated output perturbation mechanism is differentially private.  Motivated by the bias calculations in \cite{Liu16} we consider the worst case bias of these nonnegative mechanisms.   Boundary inflated truncation is generalized by considering arbitrary post-processing functions; the mechanisms are referred to as post-processed mechanisms.  We use the terminology \emph{restriction} for our generalisation of truncation (bounded) Laplace mechanisms.  The most significant contributions in this paper show that bias is inevitable for both of these generalized constructions.  We first recall the relevant definitions.

The \textbf{bias} of $\hat{Y}_q$ is given by $\mathbb{E}[\hat{Y}_q] - q$ for $q \geq 0$.  In order to analyse the question of whether some bias is inevitable for nonnegative mechanisms, we use the \emph{maximal absolute bias} of the family $\{\hat{Y}_q: q \geq 0\}$ which is defined as follows. 
\begin{definition}
\label{def:biasmax}  The maximal absolute bias of $\{\hat{Y}_q: q \geq 0\}$ is given by 
\begin{equation}\label{eq:maxbias1}
B:=\sup\{|\mathbb{E}[\hat{Y}_q] - q| : q \in [0, \infty)\}.
\end{equation}
\end{definition}

\section{Bias and nonnegative/positive mechanisms}
\label{sec:bias1}
In this section we briefly recall the core ideas of truncation and boundary inflated truncation from \cite{Liu16}, suitably adapted for nonnegative queries.  We also discuss a multiplicative mechanism from \cite{LeNyPappas13} which was introduced for strictly positive queries. 

Boundary inflated truncation works by post-processing \cite{DworkRoth14} with the standard (deterministic) ramp function 
\begin{equation}\label{eq:Tdef1}
\tau(x) = \begin{cases}
			x & \mbox{ if } x \geq 0 \\
            0 & \mbox{ otherwise.}
		\end{cases}
\end{equation}
For Laplace random variables $\{Y_q = q + L_b: q \geq 0\}$ with $b \geq \frac{\Delta}{\varepsilon}$, the mechanism corresponding to $\hat{Y}_q = \tau(Y_q)$ is $\varepsilon$-differentially private as $\tau$ is measurable \cite{HolLeiMas14, DworkRoth14}. 

It is a relatively straightforward calculation (and a limiting case of results in \cite{Liu16}) that the expectation of the random variable $\hat{Y}_q = \tau(Y_q)$ is given by 
\begin{equation}
    \label{eq:ExpPP1} \mathbb{E}[\hat{Y}_q]= q + \frac{b}{2}e^{\frac{-q}{b}}.
\end{equation}
Thus the bias of $\hat{Y}_q$ is $\frac{b}{2}e^{\frac{-q}{b}}$ and the maximal absolute bias is $\frac{b}{2}$.  By viewing boundary inflated truncation as post-processing with $\tau$, we can consider alternative functions which may give improved performance.  We see that this is indeed possible in the next section.

\textbf{Remark:} Note that for $\varepsilon$ differential privacy, we should take $b = \frac{\Delta}{\varepsilon}$; the bias of $\hat{Y}_q = \tau(Y_q)$ is then $\frac{\Delta}{2\varepsilon}e^{-\frac{q\varepsilon}{\Delta}}$.

\emph{Truncation/Restriction}

Our adaptation of truncated, or (as they are referred to in \cite{HolAntBraAon18}) bounded, Laplace mechanisms in \cite{Liu16, HolAntBraAon18} relies on the following simple result.  The proof of this is essentially identical to that used in the construction of the exponential mechanism \cite{SheTal07}; we include it here in the interests of completeness.  

\begin{proposition}\label{prop:restricted}
Let $X_{Q,d}:\Omega \to \mathbb{R}$ be an $\varepsilon$-differentially private mechanism for the query $Q:D \to [0, \infty)$.  If the family of measurable mappings $\hat{X}_{Q,d}:\Omega \to [0, \infty)$ for $d \in D$ satisfies 
\begin{equation}\label{eq:rest}\mathbb{P}(\hat{X}_{Q,d} \in A) = \frac{\mathbb{P}(X_{Q,d} \in A)}{\mathbb{P}(X_{Q,d} \in [0, \infty))}\end{equation} for all $d \in D$ and all measurable subsets $A \subseteq [0, \infty)$, then $\hat{X}_{Q,d}$ is $2 \varepsilon$-differentially private.  
\end{proposition}
\begin{proof} Let $d \sim d'$ be given and let $A \subseteq [0,\infty)$ be measurable.  Then as $X_{Q,d}$ is $\varepsilon$-differentially private:
\begin{equation}
\label{eq:rest1} \mathbb{P}(X_{Q,d} \in A) \leq e^{\varepsilon} \mathbb{P}(X_{Q,d'} \in A); \;\; \mathbb{P}(X_{Q,d'} \in [0,\infty)) \leq e^{\varepsilon} \mathbb{P}(X_{Q,d} \in [0,\infty)).
\end{equation}
Combining the two above inequalities we see that
\begin{eqnarray*}
\mathbb{P}(\hat{X}_{Q,d} \in A) &=& \frac{\mathbb{P}(X_{Q,d} \in A)}{\mathbb{P}(X_{Q,d} \in [0,\infty))}\\
&\leq& \frac{e^{\varepsilon}\mathbb{P}(X_{Q,d'} \in A)}{e^{-\varepsilon}\mathbb{P}(X_{Q,d'} \in [0,\infty))}\\
&=& e^{2\varepsilon} \mathbb{P}(\hat{X}_{Q, d'} \in A).
\end{eqnarray*} 
\end{proof}

If we start from a Laplace mechanism defined by the random variables $Y_q = q + L_b$, $q \geq 0$, $b \geq \frac{\Delta}{\varepsilon}$, the random variables defining the associated restricted mechanism satisfy
\begin{equation}\label{eq:restfamily}
\mathbb{P}(\hat{Y}_{q} \in A) = \frac{\mathbb{P}(Y_{q} \in A)}{\mathbb{P}(Y_{q} \in [0, \infty))}
\end{equation}
for Borel sets $A \subseteq [0, \infty)$ and $q \geq 0$.  It is reasonably straightforward to use this equation to obtain an expression for the distribution function and pdf of $\hat{Y}_q$.  Using these, we can show that:
\[\mathbb{E}[\hat{Y}_q] = q+\frac{q+b}{2e^{\frac{q}{b}}-1} \]
so the bias of $\hat{Y}_q$ in this case is $\frac{q+b}{2e^{\frac{q}{b}}-1}$.

\begin{corollary} \label{cor:biasrest}
For $q \geq 0$, let $Y_q = q + L_b$ be a Laplace random variable with $b = \frac{\Delta}{\varepsilon}$.  The bias of $\hat{Y}_q$ satisfying \eqref{eq:restfamily} is given by \[\frac{q \varepsilon + \Delta}{2 \varepsilon e^{\frac{q\varepsilon}{\Delta}}-\varepsilon}.\]
\end{corollary}

In order to compare mechanisms with the same guaranteed level of differential privacy, for the post-processed mechanism, we take the scale parameter $b = \frac{\Delta}{\varepsilon}$, while for the restricted Laplace mechanism, we should take $b = \frac{2 \Delta}{\varepsilon}$.  

Consider the ratio between $B_1 = \frac{\Delta}{2\varepsilon}e^{\frac{-\varepsilon q}{\Delta}}$ (the bias of a post-processed mechanism) and $B_2 = \frac{q+\frac{2 \Delta}{\varepsilon}}{2e^{\frac{\varepsilon q}{2\Delta}}-1}$ (the bias of the restricted mechanism).  After some algebraic manipulation, we find that:
\begin{eqnarray}
\frac{B_2}{B_1} = \frac{2e^{\frac{\varepsilon q}{\Delta}}}{2e^{\frac{\varepsilon q}{2 \Delta}}-1} \left(\frac{\varepsilon q}{\Delta} + 2\right) > 2.
\end{eqnarray}
This simple calculation shows that post-processing leads to a bias that is always strictly less than that caused by restriction for the Laplace mechanism.  

\subsection{Positivity and Log-Laplace Random Variables}
An alternative approach to constructing positive mechanisms was previously studied in \cite{LeNyPappas13} in order to release models of control systems in a differentially private manner.  To preserve the stability properties of the system, the mechanisms should not change the sign of certain key parameters.  \emph{Multiplicative} mechanisms based on the \emph{log-Laplace} distribution are introduced; these can be applied provided the queries and adjacency relation considered satisfy certain technical assumptions.  We note here that these mechanisms are not appropriate for the general setting we consider.  We first recall the key aspects of the setup in \cite{LeNyPappas13}.

$Q:D \to (0, \infty)$ is a positive valued query.  It is assumed that there is some $K > 0$ such that for all $d \sim d'$:
\begin{equation}
    \label{eq:LeNyBd} \frac{|Q(d) - Q(d')|}{\min \{Q(d), Q(d')\}} \leq K.
\end{equation}
This then implies that the query $\log(Q): D \to \mathbb{R}$ will have sensitivity bounded above by $K$. Hence the mechanism $X_{Q, d} = \log(Q(d)) + L_b$ is $\varepsilon$-differentially private where $L_b$ is a Laplace random variable with mean 0 and $b = \frac{K}{\varepsilon}$.  It follows that the post-processed, positive mechanism \[\hat{X}_{Q, d} = e^{X_{Q, d}} = Q(d) e^{L_b}\] is also $\varepsilon$-differentially private.  

There are several issues with using this approach for the more general setting considered here.
\begin{itemize}
    \item From the form of \eqref{eq:LeNyBd} queries cannot take the value 0.  
    \item Even when the query is strictly positive, the bound $K$ may be significantly larger than the sensitivity of the query itself, leading to noisier mechanisms.  In fact, it may not be possible to obtain a finite bound $K$.  To see this, take $D = (0, 1]^n$ and $Q: D \to (0, 1]$ given by the mean $Q(d) = \frac{\sum_{i=1}^n d_i}{n}$.  The sensitivity of $Q$ is $\frac{1}{n}$.  On the other hand, set $d = (\gamma, \gamma, \ldots, \gamma)$ and $d' = (1, \gamma, \gamma, \ldots, \gamma)$ and consider the standard adjacency relation given by $d \sim d'$ if for some $i$, $d_j = d'_j, j \neq i$.  Then \[\frac{|Q(d) - Q(d')|}{\min \{Q(d), Q(d')\}} = \frac{1-\gamma}{\gamma n}.\]
    By choosing $\gamma$ sufficiently small, we see that there is no finite $K$ for which \eqref{eq:LeNyBd} will be satisfied in this case.
    \item Following on from the last point, the next result illustrates that even if the query $Q$ satisfies \eqref{eq:LeNyBd} for some finite $K > 0$, the multiplicative mechanism will fail to have finite expectation unless $K$ satisfies additional restrictions leading to infinite bias.
\end{itemize}
\begin{proposition}
\label{prop:MultBias} Let $K > 0$, $\varepsilon > 0$  be given and let $Q:D \to (0, \infty)$ be a query satisfying \eqref{eq:LeNyBd}.  Further, let $\hat{X}_{Q, d} = Q(d) e^{L_b}$ where $L_b$ is Laplace with mean 0 and scale parameter $b = \frac{K}{\varepsilon}$.  The expected value $\mathbb{E}[\hat{X}_{Q, d}] < \infty$ if and only if $K < \varepsilon$.
\end{proposition}
\begin{proof} If $L_b$ is a Laplace random variable with mean 0 and scale parameter $b > 0$, then 
\[\mathbb{E}[e^{L_b}] = \frac{1}{2b} \int_{-\infty}^{\infty} e^x e^{-\frac{|x|}{b}} dx.\]
The integral above is finite if and only if $b < 1$.  The result follows immediately. 
\end{proof}

\textbf{Remark:} The previous result shows that even for queries satisfying \eqref{eq:LeNyBd}, the multiplicative mechanism will have an infinite bias if $K > \varepsilon$.  An identical calculation shows that $\hat{X}_{Q, d}$ will have finite variance if and only if $K < \frac{\varepsilon}{2}$.  From a practical viewpoint, these simple observations mean that for a given query satisfying \eqref{eq:LeNyBd}, it is only possible to design $\varepsilon$-differentially private mechanisms with finite mean and variance for values of $\varepsilon > 2 K$.  This is in marked contrast to mechanisms obtained by post-processing and restriction.

\section{Optimising bias over translated ramp functions}
\label{sec:BiasOpt}
Consider again a Laplace random variable $L_b$ with mean 0 and scale parameter $b > 0$.  Let $Y_q = q + L_b$ for $q \geq 0$.  For $\alpha \geq 0$, consider the translated ramp function $\tau_{\alpha}(x) = \tau(x - \alpha)$.  For $q \geq 0$, the expected value of the post-processed random variable $\tau_{\alpha}(Y_q)$ is given by
\[\mathbb{E}[\tau_{\alpha}(Y_q)] = \frac{1}{2b} \int_{\alpha}^{\infty} (x-\alpha) e^{-\frac{|x-q|}{b}} dx. \]
For a fixed $q \geq 0$, set $G(\alpha) = \mathbb{E}[\tau_{\alpha}(Y_q)]$.  

We know that for $\alpha = 0$, corresponding to the standard ramp function, the expectation of $\tau({Y_q})$ is $G(0) = q + \frac{b}{2}e^{\frac{-q}{b}}$.  This also means that the maximal absolute bias of the family $\{\tau(Y_q): q \geq 0 \}$ is $\frac{b}{2}$.  It is readily verified that the derivative of $G$ with respect to $\alpha$ is given by $G'(\alpha) = -\frac{1}{2b} \int_{\alpha}^{\infty} e^{-\frac{|x-q|}{b}} dx < 0$.  Thus for any fixed $q$, $G(\alpha)$ is a decreasing function of $\alpha$.  Moreover, for every $q$, the bias of $\tau(Y_q)$ is strictly positive, meaning that $G(0) > 0$.  These observations suggest that it may be possible to reduce the bias of $\tau(Y_q)$ by instead considering $\tau_{\alpha}(Y_q)$ for $\alpha > 0$. 

For $\alpha \geq 0$, let $B(\alpha)$ denote the maximal absolute bias \eqref{eq:maxbias1} of the family $\{\tau_{\alpha}(Y_q): q \geq 0\}$.  In the following result, we determine the minimum value of $B(\alpha)$ over $\alpha \geq 0$.  

\begin{theorem}
\label{thm:minalpha} Let $L_b$ be a Laplace random variable with mean 0 and scale parameter $b > 0$ and, for $q \geq 0$, let $Y_q = q + L_b$.  Let $\alpha^*$ be the unique solution of $\frac{b}{2}e^{-\frac{\alpha}{b}} -\alpha = 0$ in $[0, \infty)$.  Then $\min\{B(\alpha) : \alpha \geq 0\} = \alpha^*$.
\end{theorem}
\begin{proof} Fix some $\alpha \geq 0$.  It can be readily verified by direct calculation that the expectation of $\tau_{\alpha}(Y_q)$ is given by 
\[\mathbb{E}[\tau_{\alpha}(Y_q)] = (q-\alpha) + \frac{b}{2}e^{\frac{\alpha - q}{b}} \]
for $q \geq \alpha$ and 
\[\mathbb{E}[\tau_{\alpha}(Y_q)] = \frac{b}{2}e^{\frac{q -\alpha}{b}} \]
for $0 \leq q \leq \alpha$.  
It follows that the bias $\beta_{\alpha}(q) = \mathbb{E}[\tau_{\alpha}(Y_q)] - q$ is given by:
\begin{itemize}
    \item $\beta_{\alpha}(q) = - \alpha + \frac{b}{2}e^{\frac{\alpha - q}{b}}$ for $q \geq \alpha$;
    \item $\beta_{\alpha}(q) = \frac{b}{2}e^{\frac{q -\alpha}{b}} - q$ for $0 \leq q \leq \alpha$.
\end{itemize}
It is clear that $\beta_{\alpha}(q)$ is a monotonically decreasing function of $q$ for $q \geq \alpha$.  Moreover, a simple calculation shows that for $0 \leq q \leq \alpha$, $\beta_{\alpha}'(q) = \frac{1}{2} e^{\frac{q-\alpha}{b}} - 1$ which is negative for $0\leq q \leq \alpha $.  Thus $\beta_{\alpha}(q)$ is a decreasing function of $q \geq 0$ for all $\alpha \geq 0$ and moreover it is continuous.  This implies that the maximum absolute bias of $\{\tau_{\alpha}(Y_q): q \geq 0\}$ is either given by $|\beta_{\alpha}(0)| = \frac{b}{2}e^{-\frac{\alpha}{b}}$ or the limit $\lim_{q \to \infty} |\beta_{\alpha}(q)| = \alpha$.  Formally, 
\[B(\alpha) = \max\{\frac{b}{2}e^{-\frac{\alpha}{b}}, \alpha\}.\]
To complete the proof, let $\alpha^*$ be the unique solution of $\frac{b}{2}e^{-\frac{\alpha}{b}} -\alpha = 0$ in $[0, \infty)$.  Then:
\begin{itemize}
    \item[(i)] $B(\alpha) = \frac{b}{2}e^{-\frac{\alpha}{b}}$ on $[0, \alpha^*]$;
    \item[(ii)] $B(\alpha) = \alpha$ on $[\alpha^*, \infty)$.
\end{itemize}
Clearly $B(\alpha)$ is decreasing on $[0, \alpha^*]$ and increasing on $[\alpha^*, \infty)$.  Hence the minimum value of $B(\alpha)$ on $[0, \infty)$ is given by $B(\alpha^*) = \alpha^*$ as claimed.  
\end{proof}

\textbf{Remark:} As $\alpha^* > 0$, the maximal absolute bias of $\{\tau_{\alpha^*}(Y_q): q \geq 0\}$, given by $\frac{b}{2} e^{-\frac{\alpha^*}{b}}$ is clearly less than $\frac{b}{2}$ corresponding to the standard ramp function.

\section{Bias is inevitable for post-processing and restriction}
\label{sec:biasinev}
In Section \ref{sec:bias1}, we noted that the maximal absolute bias of nonnegative mechanisms constructed by either restriction or post-processing with the ramp function is strictly positive.  In this section, we prove that this is a fundamental property of \textbf{any} post-processed, nonnegative Laplace mechanism and any nonnegative restricted mechanism (irrespective of what the original mechanism is).  

\subsection{Bias and post-processed Laplace mechanisms}
We first consider the maximal absolute bias for post-processed Laplace mechanisms.  Throughout this subsection, $\{Y_q: q \geq 0\}$ is a set of Laplace random variables with scale parameter $b > 0$; each $Y_q$ has the pdf given by \eqref{eq:Lap}.  Let $\phi:\mathbb{R} \to [0, \infty)$ be a measurable function and define the post-processed family $\hat{Y}_{\phi, q}$ by $\hat{Y}_{\phi, q} = \phi(Y_q)$ for $q \geq 0$.  

In order to ensure that each $\hat{Y}_{\phi, q}$ $q\geq 0$, has finite first and second moments, we consider post-processing functions in the Hilbert space $V = L^2(e^{-\frac{|x|}{b}}dx)$:  
\[V : = \{\phi:\mathbb{R} \to \mathbb{R}: \int_{-\infty}^{\infty} |\phi(x)|^2 e^{-\frac{|x|}{b}}dx < \infty \}.\]
The cone of nonnegative valued functions $\phi$ in $V$ is denoted by $V_+$.  
It is straightforward to show that $\phi$ in $V$ implies that $\int_{-\infty}^{\infty} |\phi(x)| e^{-\frac{|x|}{b}}dx < \infty$ also.  In the following lemma we note that given $\phi \in V_+$, the first and second moments of $\hat{Y}_{\phi, q}$ are finite for all $q \geq 0$.  
\begin{lemma}
\label{lem:finitemoments}
Let $\phi \in V_+$ be given.  Then for any $q \geq 0$:
\begin{eqnarray*}
\int_{-\infty}^{\infty} |\phi(x)| e^{-\frac{|x-q|}{b}} dx < \infty, \;
\int_{-\infty}^{\infty} |\phi(x)|^2 e^{-\frac{|x-q|}{b}} dx < \infty.
\end{eqnarray*} 
\end{lemma}
\begin{proof} As $\phi \in V$, we know that 
\begin{equation}
\label{eq:finmom1} \int_{-\infty}^{\infty} |\phi(x)|^p e^{-\frac{|x|}{b}} dx  < \infty
\end{equation}
for $p = 1, 2$.  The result now follows from a simple application of the triangle inequality as 
\[e^{-\frac{|x-q|}{b}} \leq e^{\frac{|q|}{b}} e^{-\frac{|x|}{b}}.\]
\end{proof}

We now consider the mapping from $V$ into $\mathbb{R}$ which takes a function $\phi$ to the maximal absolute bias given by \eqref{eq:maxbias1}.  Formally, for $\phi \in V$:
\begin{equation}
\label{eq:maxbiasphi} B(\phi) = \sup \left\{\left|\int_{-\infty}^{\infty} \phi(x) f_q(x)dx - q\right|: q \in [0, \infty)\right\}.
\end{equation}  
The following result establishes that $B(\phi) > 0$ for any $\phi \in V_+$.  
\begin{proposition}
\label{prop:biasinevpp}
Let $\phi \in V_+$ be given and let $B(\phi)$ be defined by \eqref{eq:maxbiasphi}.  Then $B(\phi) > 0$.  
\end{proposition}
\begin{proof} Consider $q = 0$.  Then as $\phi \in V_+$ and $f_0(x) > 0$ for all $x$, it follows that 
\[\int_{-\infty}^{\infty} \phi(x) f_0(x) dx  > 0\]
unless $\phi = 0$ almost everywhere.  If $\phi$ is not zero a.e., it follows immediately that
\[B(\phi) \geq \int_{-\infty}^{\infty} \phi(x) f_0(x)dx  > 0.\] 
On the other hand, if $\phi = 0$ a.e. then
\[\left|\int_{-\infty}^{\infty} \phi(x) f_q(x)dx - q\right| = q\]
for all $q \geq 0$ which means that $B(\phi) = \infty$ in this case.  
\end{proof}

\textbf{Remark:}  The last result shows that the maximal absolute bias, $B(\phi)$, of any post-processed Laplace mechanism (with finite first and second moments) must be positive.  In the next result, we establish the stronger fact that the maximal absolute bias of such mechanisms is bounded away from zero; formally $\inf\{B(\phi): \phi \in V_+\} > 0$.  

\begin{proposition}\label{prop:infbiaspos}
Let $B(\phi)$ be given by \eqref{eq:maxbiasphi} for $\phi \in V_+$.  Then 
\[\inf\{B(\phi): \phi \in V_+\} > 0.\]
\end{proposition} 
\begin{proof} We argue by contradiction.  If the infimum was equal to 0, there would exist some sequence of functions $\phi_n$ in $V_+$ such that $B(\phi_n) < \frac{1}{n}$ for all $n$.  From the definition of $B(\phi)$, this would mean that for all $q \in [0, \infty)$ and all $n \geq 0$, 
\begin{equation}\label{eq:infbias1}
\left|\int_{-\infty}^{\infty} \phi_n(x) f_q(x) dx - q \right| < \frac{1}{n}.
\end{equation}
For $q = 0$ this implies that for all $n$, $\int_{-\infty}^{\infty} \phi_n(x) f_0(x) dx < \frac{1}{n}$.  Now note that for $q \in [0, \infty)$, 
\begin{eqnarray*}
f_q(x) = \frac{1}{2b} e^{-\frac{|x-q|}{b}} \leq  e^{\frac{q}{b}} f_0(x) .
\end{eqnarray*}
This implies that for all $q \in [0, \infty)$, $n \geq 0$:
\begin{eqnarray*}
\int_{-\infty}^{\infty} \phi_n(x) f_q(x) dx &\leq&  e^{\frac{q}{b}} \int_{-\infty}^{\infty} \phi_n(x) f_0(x) dx \\
&\leq& \frac{e^{\frac{q}{b}}}{n}.
\end{eqnarray*}
We can now use this and $\phi_n \in V_+$ to conclude that for all $q, n$:
\begin{eqnarray*}
\left|\int_{-\infty}^{\infty} \phi_n(x) f_q(x) dx - q \right| &\geq& q - \int_{-\infty}^{\infty} \phi_n(x) f_q(x) dx \\
&\geq& q - \frac{e^{\frac{q}{b}}}{n}.
\end{eqnarray*}
This clearly contradicts \eqref{eq:infbias1} as together they would imply that for all $q, n$
\[q < \frac{1}{n}(e^{\frac{q}{b}}+ 1)\]
which is impossible.  
\end{proof}

\subsection{Bias and restricted mechanisms}
We now consider the maximal absolute bias of general restricted mechanisms.  Let a family of continuous real-valued random variables $\{Y_q: q \geq 0\}$ with associated pdfs $f_q$, $q \geq 0$ be given.  We make the following assumptions for all $q \geq 0$:
\begin{eqnarray}
\label{eq:restasp1} \mathbb{E}[Y_q] &=& q\\
\label{eq:restasp2} f_q(x) &>& 0 \;\;\;\; \forall x \in \mathbb{R}.
\end{eqnarray}
Thus, we are assuming that the base mechanism is unbiased and has range given by $\mathbb{R}$.  

Let $\hat{Y}_q$ denote the restricted family of nonnegative random variables satisfying \eqref{eq:restfamily}.

We will show that for $\hat{Y}_q$, the maximal bias given by \eqref{eq:maxbias1} also satisfies $B > 0$.  The proof of this makes use of the standard coupling technique from probability theory (see Section 4.12 of \cite{GriSti}) to construct a common space $\Omega_1$ on which $\hat{Y}_q$ and a copy of $Y_q$ can be defined for all $q \geq 0$. 

\begin{proposition}
\label{prop:biasinevrest} Let a family of random variables $\{Y_q: q \geq 0\} $ satisfying \eqref{eq:restasp1}, \eqref{eq:restasp2} and a restricted family $\hat{Y}_q$ satisfying \eqref{eq:restfamily} be given.  Then the maximal absolute bias $B$ given by \eqref{eq:maxbias1} satisfies $B > 0$.
\end{proposition}
\begin{proof} For $q \geq 0$, let $F_q$ denote the cumulative distribution function of $Y_q$ and $F^R_q$ the cdf of $\hat{Y}_q$.  As $\hat{Y}_q$ takes values in $[0, \infty)$, $F^R_q(t) = 0$ for $t < 0$.  Moreover, it follows from \eqref{eq:restfamily} that for $t \geq 0$:
\begin{eqnarray*}
F^R_q(t) &=& \mathbb{P}(\hat{Y}_q \leq t) \\
&=& \frac{F_q(t) - F_{q}(0)}{1 - F_{q}(0)}. 
\end{eqnarray*}    
We now show that $F^R_q(t) < F_q(t)$ for all $t \in \mathbb{R}$.  This is trivially true for $t < 0$.  For $t \geq 0$, 
\begin{eqnarray*}
F^R_q(t) &=& \frac{F_q(t) - F_{q}(0)}{1 - F_{q}(0)}  \\
&=& \frac{F_q(t)(1-F_{q}(0)) + F_{q}(0)(F_q(t) - 1)}{1 - F_{q}(0)}\\
&=&F_q(t) + F_{q}(0)\frac{(F_q(t) - 1)}{1 - F_{q}(0)}.
\end{eqnarray*}    
As $f_q(x) > 0$ for all $x$ in $\mathbb{R}$, it follows that $F_q(t) < 1$ for all $t \in \mathbb{R}$ and $F_{q}(0) < 1$.  This immediately implies that $F^R_q(t) < F_q(t)$ for $t \geq 0$ also.  Therefore, the restricted mechanism $\hat{Y}_q$ stochastically dominates $Y_q$ \cite{GriSti} for all $q \geq 0$.  This means that we can construct a probability space $(\Omega_1, \mathcal{F}_1, \mathbb{P}_1)$ and random variables $\hat{Y}^1_q$, $Y^1_q$ for $q \geq 0$ such that   
\begin{enumerate}
\item $\hat{Y}^1_q$ has the same distribution (cdf) as $\hat{Y}_q$ and $Y^1_q$ has the same distribution as $Y_q$ for all $q \geq 0$;
\item $\hat{Y}^1_q(\omega) \geq Y^1_q(\omega)$ for all $q \geq 0$ and all $\omega \in \Omega_1$. 
\end{enumerate}
We make a slight adaptation of a standard construction in order to strengthen statement 2 slightly and prove that $\hat{Y}^1_q$ has strictly positive bias.   We set $\Omega_1 = [0, 1]$, take $\mathcal{F}_1$ to be the Borel subsets of $[0, 1]$, and define $\mathbb{P}_1$ to be the Lebesgue measure on $[0, 1]$.  For $q \geq 0$, we define random variables $\hat{Y}^1_q$, $Y^1_q$ on $\Omega_1$ by setting:
\begin{eqnarray*}
\hat{Y}_q^1(\omega) &=& \inf\{t: F^R_q(t) > \omega\}\\
Y^1_q(\omega) &=& \inf\{t:F_q(t) > \omega\}.
\end{eqnarray*}
As each of the cdfs, $F^R_q$, $F_q$ for $q \geq 0$ is continuous and non-decreasing, it is not difficult to see that for $t \in \mathbb{R}$, $\hat{Y}_q^1(\omega) \leq t \Leftrightarrow \omega \leq F^R_q(t)$, and $Y^1_q(\omega) \leq t \Leftrightarrow \omega \leq F_q(t)$.  These two facts imply that 
\[\mathbb{P}(\hat{Y}^1_q \leq t) = F^R_q(t); \mathbb{P}(Y^1_q \leq t) = F_q(t).\]
It follows immediately that as $\hat{Y}^1_q$ has the same distribution as $\hat{Y}_q$, $\mathbb{E}[\hat{Y}^1_q] = \mathbb{E}[\hat{Y}_q]$.  Similarly, $\mathbb{E}[Y^1_q] = \mathbb{E}[Y_q] = q$.  Furthermore, as $F^R_q(t) < F_q(t)$ for all $t \in \mathbb{R}$, $\hat{Y}_q^1(\omega) \geq Y^1_q(\omega)$ for all $\omega \in [0, 1]$.  We next show that for every $q \geq 0$, there exists some subset $S_q$ of $(0, 1)$ of positive measure with the property that 
\[\hat{Y}_q^1(\omega) > Y^1_q(\omega) \;\; \forall \omega \in S_q.\]

As $F_q(t) = \int_{-\infty}^t f_q(x)dx$ with $f_q(x) > 0$ for $x \in (-\infty, \infty)$ by assumption, it follows that we can choose $K > 0$ and $\alpha > 0$ such that $F_q(-K) = \alpha$.  This immediately implies that $Y^1_q(\omega) \leq -K$ for $\omega \in (0, \alpha)$.  Moreover, by construction $\hat{Y}^1_q(\omega) \geq 0$ for all $\omega \in (0, 1)$ and hence taking $S_q = (0, \alpha)$, we have that 
\[\hat{Y}^1_q(\omega) - Y^1_q(\omega) \geq K, \;\; \forall \omega \in S_q.\]
It now follows immediately that: 
\begin{eqnarray*}
\mathbb{E}[\hat{Y}_q] &=& \mathbb{E}[\hat{Y}^1_q] \\
&=& \int_{\Omega_1} \hat{Y}^1_q(\omega) d\mathbb{P}_1(\omega) \\
&=& \int_{S_q} \hat{Y}^1_q(\omega) d\mathbb{P}_1(\omega) + \int_{S_q^c} \hat{Y}^1_q(\omega) d\mathbb{P}_1(\omega)\\
&>& \int_{S_q} Y^1_q(\omega) d\mathbb{P}_1(\omega) + \int_{S_q^c} Y^1_q(\omega) d\mathbb{P}_1(\omega)\\
&=& \mathbb{E}[Y^1_q] = \mathbb{E}[Y_q].
\end{eqnarray*}
Here $S_q^c = \Omega_1 \setminus S_q$. The argument above shows that, for any $q \geq 0$ $\mathbb{E}[\hat{Y}_q] > \mathbb{E}[Y_q] = q$ and hence the maximal absolute bias $B$ satisfies:
\[B = \sup\{|\mathbb{E}[\hat{Y}_q] - q| : q \geq 0\} > 0.\]
\end{proof}

\section{Conclusions and Discussion}
\label{sec:conc}
An advantage of viewing boundary inflated truncation (to use the terminology of \cite{Liu16}) as post-processing with the ramp function, $\tau$, is that this opens up the possibility of using alternative post-processing functions in order to reduce bias for nonnegative, post-processed mechanisms.  The results presented in Section \ref{sec:BiasOpt} show that this is indeed possible even by using simple translations of the ramp function.  We have given an explicit characterisation of the optimal post-processing function within this class of functions.  The work of Section \ref{sec:biasinev} proves that the maximum absolute bias of any nonnegative post-processed mechanism must be strictly positive.  We have also derived a corresponding result for restricted mechanisms constructed from any initial mechanism.  

The results here suggest a number of directions for future work.  The result of Proposition \ref{prop:infbiaspos} shows that the worst case bias of any nonnegative post-processed Laplace mechanism is positive.  This means that the infimal or minimal value over all post-processing functions is positive.  An interesting problem is to quantitatively characterise this infimum and determine whether it is attained by some post-processing function (in essence, derive a theorem guaranteeing the existence of a minimiser).  If such a function exists, providing an explicit, or implicit, characterisation of it would also be interesting, both practically and theoretically.  The more general question of characterising the optimal mechanism for nonnegative queries is a natural extension of this line of work.  Corresponding questions for the restriction-based mechanisms can also be investigated.  

All of our contributions concern bias and it is important to also consider other performance metrics such as the MSE.  The authors have done some initial work in this direction and hope to report results of this nature in the near future.  Finally, it would be very interesting to study how our results for unbounded nonnegative queries relate to those of \cite{Cal20} and \cite{FioHen20B} for integer-valued and bounded constrained queries.


\section*{Acknowledgments}
The authors would like to express their sincere gratitude to the editors and anonymous reviewers for their careful reading of the manuscript.  We are particularly thankful for the thoughtful and insightful suggestions; these have not only helped to improve the paper but have also highlighted some interesting related work.

This work was supported by the Science Foundation Ireland Grant 13/RC/2094, and the European Regional Development Fund through the Southern \& Eastern Regional Operational Programme to Lero -- the Irish Software Research Centre (www.lero.ie).


\medskip
Received xxxx 20xx; revised xxxx 20xx.
\medskip


\begin{thebibliography}{99}

\bibitem{Abowd}
	\newblock J. M. Abowd,
	\newblock The U.S. Census Bureau adopts differential privacy,
	\newblock in \emph{Proceedings of the 24th ACM SIGKDD International Conference on Knowledge Discovery \& Data Mining} (2018): 2867, DOI: 10.1145/3219819.3226070.
	
\bibitem{DomSor15}
     \newblock  J. Domingo-Ferrer and J. Soria-Comas,  
     \newblock From t-closeness to differential privacy and vice versa in data anonymization,
     \newblock \emph{Knowledge Based Systems}, \textbf{74} (2015), 151--158.

\bibitem{Dwork06A} 
   \newblock C. Dwork, F. McSherry, K. Nissim and A. Smith,
   \newblock Calibrating Noise to Sensitivity in Private Data Analysis,
   \newblock \emph{Theory of Cryptography}, \textbf{3876} (2006), 265--284.
  
\bibitem{Dwork06} 
	\newblock C. Dwork,  
     \newblock Differential privacy,
     \newblock in \emph{Proceedings of the 33rd Annual International Colloquium on Automata, Languages and
Programming, LNCS 4051}, Springer-Verlag (2006), 1--12.

\bibitem{DworkRoth14}
	\newblock C. Dwork and A. Roth, 
	\newblock The algorithmic foundations of differential privacy, 
	\newblock \emph{Foundations and Trends in Theoretical Computer Science}, \textbf{9} (2014), 211--407.

\bibitem{FioHen20A}
	\newblock F. Fioretto, P. Van Hentenryck and K. Zhu, 
	\newblock Differential privacy of hierarchical census data: an optimization approach,
	\newblock arXiv:2006.15673 (2020)

\bibitem{GriSti} 
	\newblock G. Grimmett and D. Stirzaker, 
	\newblock \emph{Probability and Random Processes},
	\newblock Oxford University Press, 2001.

\bibitem{HalRinWas12} 
	\newblock R. Hall, A. Rinaldi,  and L. Wasserman, 
	\newblock Random differential privacy, 
	\newblock \emph{Journal of Privacy and Confidentiality} \textbf{2} (2012), 43--59.
	
\bibitem{HolAntBraAon18}
	\newblock N. Holohan, S. Antonatos, S. Braghin, and P. Mac Aonghusa, 
	\newblock The bounded Laplace mechanism in differential privacy, 
	\newblock \emph{Journal of Privacy and Confidentiality}, \textbf{10} (2020), https://doi.org/10.29012/jpc.715.
	
\bibitem{HolLeiMas14} 
	\newblock N. Holohan, D. Leith, and O. Mason, 
	\newblock Differential privacy in metric spaces: Numerical, categorical and functional data under the one
roof, 
	\newblock \emph{Information Sciences}, \textbf{305} (2015), 256--268.
	
\bibitem{HolLeiMas17}
	\newblock N. Holohan, D. Leith, and O. Mason, 
	\newblock Optimal differentially private mechanisms for randomised response,
	\newblock \emph{IEEE Transactions on Information Forensics and Security}, \textbf{12} (2017), 2726--2735.

\bibitem{KalSanSar18} 
	\newblock K. Kalantari, L. Sankar, and A. Sarwate, 
	\newblock Robust privacy-utility tradeoffs under differential privacy and hamming distortion, 
	\newblock \emph{IEEE Transactions on Information Forensics and Security}, \textbf{13} (2019), 2816--2829. 

\bibitem{LeNyPappas13} 
	\newblock J. Le Ny and G. J. Pappas, 
	\newblock Privacy-preserving release of aggregate dynamic models, 
	\newblock in \emph{Proceedings of HiCoNS}, (2013).

\bibitem{Liu16}
	\newblock F. Liu, 
	\newblock Statistical properties of sanitized results from differentially private Laplace mechanism with bounding constraints, 	\newblock preprint \arXiv{stat.ME/1607.08554}
	
\bibitem{Liu19}
	\newblock F. Liu,
	\newblock Generalized gaussian mechanism for differential privacy, 
	\newblock \emph{IEEE Transactions on Knowledge and Data Engineering,} \textbf{21} (2019), 747--756.

\bibitem{LiuKarZha18}
	\newblock L. Liu, H. Karimi, and X. Zhao, 
	\newblock New approaches to positive observer design for discrete-time positive linear systems, 
	\newblock \emph{Journal of the Franklin Institute}, \textbf{355} (2018), 4336--4350.

\bibitem{SheTal07}
	\newblock F. McSherry and K. Talwar,
	\newblock Mechanism design via Differential Privacy,
	\newblock in \emph{Proceedings of 48th Annual Symposium of Foundations of Computer Science} (2007), 94--103

\bibitem{Cal20}
	\newblock P. Sadeghi, S. Asoodeh and F. du Pin Calmon,
	\newblock Differentially private mechanisms for count queries,
	\newblock arXiv:2007.09374 (2020)

\bibitem{SorDom13}
	\newblock J. Soria-Comas and J. Domingo-Ferrer, 
	\newblock Optimal data independent noise for differential privacy, 
	\newblock \emph{Information Sciences}, \textbf{250} (2013), 200--214.

\bibitem{SorDomSanMeg17} 
	\newblock J. Soria-Comas, J. Domingo-Ferrer, D. Sanchez, and D. Megias, 
	\newblock Individual differential privacy: A utility-preserving formulation of differential privacy guarantees,
	\newblock \emph{IEEE Tran. Information Forensics and Data Security}, \textbf{12} (2017), 1418--1429.

\bibitem{Torra17} 
	\newblock V. Torra,
	\newblock \emph{Data Privacy: Foundations, New Developments and the Big Data Challenge}, 
	\newblock Springer-Verlag, 2017.

\bibitem{Bayes1}
	\newblock A. Triastcyn and B. Faltings,
	\newblock Bayesian differential privacy for machine learning,
	\newblock in \emph{Proceedings of the 37th International Conference on Machine Learning}, (2020), 9583--9592
		
\bibitem{ValRan2018} 
	\newblock E. Valcher and A. Rantzer, 
	\newblock A tutorial on positive systems and large scale control,
	\newblock in \emph{Proceedings of IEEE Conf. on Dec. and Cont.}, (2018).

\bibitem{FioHen20B}
	\newblock K. Zhu, P. Van Hentenryck and F. Fioretto,
	\newblock Bias and variance of post-processing in differential privacy,
	\newblock arXiv:2010.04327 (2020)
		
\end{thebibliography}
\end{document}